\documentclass[12pt,a4paper]{amsart}

\usepackage{geometry}
  \usepackage{amsmath}
\usepackage{amssymb}
\usepackage{amsfonts}
\usepackage{amsbsy,amsthm}
\usepackage{latexsym}

\newtheorem{lemma}{Lemma}
\newtheorem{theorem}{Theorem}

\begin{document}

\thispagestyle{empty}

\title[On weakly commutative triples]{On weakly commutative triples of partial differential operators}
\author{{S.P. Tsarev, V.A. Stepanenko}}%
\address {Sergey P. Tsarev, Vitaly A. Stepanenko
\newline \hphantom {iii} Siberian Federal University,
\newline \hphantom {iii} pr. Svobodny, 79,
\newline \hphantom {iii} 660041, Krasnoyarsk, Russia}%
\email{sptsarev@mail.ru}%

\thanks{\copyright \ 2017, Tsarev S.P., Stepanenko V.A.}
\thanks{\rm The work is supported by the RNF grant (project 14-11-00441)}


\maketitle 
{\small
\begin{quote}
\noindent{\sc Abstract. } We investigate algebraic properties of weakly commutative triples, appearing in
the theory of integrable nonlinear partial differential equations. 
Algebraic technique of skew fields of formal pseudodifferential 
operators as well as skew Ore fields of fractions are applied to this problem, relating weakly commutative 
triples to commuting elements of skew Ore field of formal fractions of ordinary differential operators.
A version of Burchnall-Chaundy theorem for weakly commutative triples is proved by algebraic means avoiding 
analytical complications typical
for its proofs known in the theory of integrable equations.
\medskip

\end{quote}
}

\section{Introduction}

In the theory of integrable nonlinear equations with three independent variables, an algebro-geometric
theory similar to the theory of commuting ordinary differential operators was often used 
(cf., for example, \cite{DKN, K97}). Namely one considers 
\cite{K97} two
operator $ L_1 $, $ L_2 $ and the Schr\"odinger operator $ H = - \partial_x^2 - \partial_y^2 + u (x, y) $ 
in the ring of differential operators $ F [\partial_x, \partial_y] $ with variable coefficients, such that
\begin{equation}\label{Kri1}
[L_1, L_2] = P_0 \cdot H, \quad [L_1, H] = P_1 \cdot H, \quad
[L_2, H] = P_2 \cdot H
\end{equation}
for some $ P_0, P_1, P_2 \in F[\partial_x, \partial_y] $. 
Such triples of operators will be called below 
triples of weakly commutative operators (commuting $ \mod H $) or, more specifically, triples of 
$H$-commuting operators. For simplicity, we will work with the hyperbolic version of the identities 
(\ref{Kri1}), in which we set $ H = - \partial_x \partial_y + u $. 
Such triples have many remarkable properties, of which we give below the following analog of the Burchnall-Chaundy theorem 
\cite{BC, K97}:
\begin{theorem}\label{thBurCha3}
 The operators $ L_1 $ and $ L_2 $ satisfying the equations (\ref{Kri1}) satisfy a polynomial
 relation $ Q (L_1, L_2) = 0 $ with constant coefficients on the solution space $ H = 0 $, in other words,
 $ Q ( L_1, L_2)  = S \cdot H $, $ S \in 
 F [\partial_x, \partial_y] $.
\end{theorem}
In addition to the rich analytical and algebro-geometric structures,
this theory has an important purely algebraic aspect, to which this work is devoted. 
Namely, we study formal algebraic properties 
of triples (\ref{Kri1}) and prove the algebraic analogues of the analytic results of \cite{DKN, K97} 
without resorting to many subtle analytic details
(implicitly contained, for example, in the theorem above). 
In particular, we prove Theorem~\ref{thBurCha3} by purely algebraic means. In addition, we give other simple algebraic properties of weakly commutative triples. 
Generalizing the technique of Schur~\cite{Schur}, we show 
the commutativity modulo $ H $ of all operators $ L_2 $ weakly
commuting with given $ H $ and $ L_1 $ and the existence of 
a formal analytic (generally speaking non-polynomial) relation (a series with constant coefficients) 
between commuting formal pseudodifferential ordinary operators from the 
respective Ore skew field, therefore partially solving one of the problems posed in~\cite{Goodearl}. 

In the next Section we give the necessary information from the algebraic theory 
of formal pseudodifferential operators and an overview of the constructions and results we need. 

In Section~\ref{sec-primery} we describe the simplest properties 
of weakly commutative triples and give useful examples of families of operators weakly commuting with $ H $. 

Section~\ref{sec-our-th} is devoted to the proof of the main results of this paper. 

In the Conclusion we discuss the unsolved algebraic problems in the theory of commuting pseudodifferential operators. 


\section{The skew field of formal pseudodifferential operators and the Ore skew field of formal fractions of differential operators} 

In what follows we often use the operation of formal inversion of ordinary and partial differential operators. 
The skew field of pseudodifferential operators of one variable, traditional for the theory of integrable nonlinear differential equations, is defined in our context as a skew field of formal series of the form 
\begin{equation}\label{pdo1}
  L=\sum_{i=0}^\infty
  p_{n-i}X^{n-i}=p_nX^n+p_{n-1}X^{n-1}+p_{n-2}X^{n-2}+\ldots  
\end{equation}
where $ p_k $ are smooth or analytic functions of two variables $ x $, $ y $. Here and below, we use the notation 
$ X = \frac{\partial}{\partial x} $, 
$ Y = \frac{\partial}{\partial y} $.
However, it is not always reasonable to work with
pseudodifferential operators in the form of infinite series. 
Often a very simple and constructive definition of the ``Ore skew field of fractions'' is more convenient. 
This construction is well known in non-commutative algebra, it is
analogous to the construction of the field of fractions for an integral domain in commutative 
algebra (see \cite{Ore, Cohn, Goodearl}). Briefly, the construction of the Ore skew field is done
as follows: let us consider formal elements of the form $ L^{-1} \cdot M $ or $ B \cdot A^{-1} $ (each of them can be rewritten in the other form) where $ A $, $ B $, $ L $, $ M $ are differential operators (ordinary or partial); 
in order to determine the operations of addition of such formal fractions, it is necessary to reduce them to a common denominator, 
finding a common multiple (not necessarily the least common multiple) of the denominators. 
A brief description of the Ore construction and the necessary conditions for its correctness can be found in \cite{TT10}. 
We denote by $ F (\partial_x) \equiv F (X) $ the result of applying of this construction to the ring $ F [\partial_x] \equiv F [X] $ of ordinary differential operators with coefficients in 
some function space $ F $, by analogy with by construction of the field of rational functions $ {\mathbb Q} (x) $ from the polynomial ring $ {\mathbb Q} [x] $. 
In the following sections, we always use the field of (locally) analytic functions of two variables $ \mathcal{F}_2 = \{f (x, y) \} $ as the field $ F $ of coefficients. 
Similarly, the result of applying the Ore construction to the
ring of partial differential operators $ \mathcal{F}_2 [\partial_y, \partial_x] $ with derivatives with respect to $ x, y $ will be denoted by $ \mathcal{F}_2 (\partial_y, \partial_x)  \equiv \mathcal{F}_2 (X, Y) $. 
Below we often use the ring $ \mathcal{F}_2 (\partial_x) [\partial_y] \equiv \mathcal{F}_2 (X) [Y] $ of formal ordinary differential operators with respect to $ \partial_y $ 
with the Ore skew field $ \mathcal{F}_2 (\partial_x) $ of coefficients. 
All the algebraic properties of the ring of ordinary differential operators are preserved (in fact, as already noted by O.Ore, 
it is possible to construct the rings of differential operators and the skew fields of their fractions over non-commutative differential fields). 
The skew field of formal series (\ref{pdo1}) will be denoted by $ \mathcal{F}_2 ((\partial_x)) \equiv \mathcal{F}_2 ((X)) $. 
For reference, we give the explicit expression for the multiplication in this skew field 
(the composition of formal operators): for $ L_1 $ of the form (\ref{pdo1}) and 
$$ L_2 = \sum_{i = 0}^\infty q_{m-i} X^{m-i } = q_mX^m + q_{m-1} X^{m-1} + q_{m-2} X^{m-2} + \ldots, 
$$ 
we have 
\begin{equation}
  L_1\cdot
  L_2=p_nq_mX^{n+m}+\sum^\infty_{k=1} \left ( \sum_{j=0}^k
  p_{n-j}\left ( \sum^{k-j}_{i=0} {n-j\choose k-j-i}
  q_{m-i}^{(k-i-j)}\right) \right)X^{m+n-k}.
\end{equation}
In this formula the superscript $(k-i-j)$ above the coefficient $ q_{m-i } $ denotes its derivative of order $ k-i-j $ with respect to the variable $ x $. 

For all of the skew fields and rings of formal operators considered here, there is a natural extension of the order of an element 
that extends the order of a differential operator.

In fact many results and ideas of the theory of formal pseudodifferential operators of the form (\ref{pdo1}) were given
by I. Schur~\cite{Schur}. 
In particular, he proved the following important theorems. 
\begin{theorem}\label{thSch-2} 
For any element $ L \in \mathcal{F}_1 ((\partial_x)) $ (with coefficients from the field of smooth functions of one variable)
of arbitrary non-zero order $ n $, there exists a root 
$ R = \sqrt [n] {P} $ of degree $ n $, thus having order 1. 
\end{theorem} 
Actually I.Schur proves the existence of an element
$ R \in \mathcal{F}_1 ((\partial_x)) $ such that 
$ R^n = P $; for each choice of one of the $ n $ roots, $ r_1 = \sqrt [n] {p_n} $ ($ p_n $ is the leading coefficient of $ L $, see (\ref{pdo1})) as the leading coefficient of the  
root $ R $, all the other coefficients of $ R $
are uniquely recovered. 
\begin{theorem}\label{thSch-3} 
If two elements $ L_1, L_2 \in \mathcal{F}_1 ((\partial_x)) $ (of orders $ n $ and $ m \neq 0 $ respectively) commute, then there exist constants $ c_0 $, $ c_1 $, $ c_2, \ldots, $ such that%
\begin{equation}\label{eq-ryad-ci-Schur}
  L_1=\sum^\infty_{k=0}c_kL_2^{\frac{n-k}{m}}.
\end{equation} 
\end{theorem} 
The following main statement of the paper~\cite{Schur} easily follows from the previous two theorems: 
\begin{theorem}\label{thSch-1} 
  If two elements $ L_1, L_2 \in \mathcal{F}_1 ((\partial_x)) $
  commute with the third $ L_3 \in \mathcal{F}_1 ((\partial_x)) $ that is not equal to a constant, then $ L_1 $ and $ L_2 $ commute with each other.
\end{theorem}

We actually reproduce the proofs of these three theorems in 
Section~\ref{sec-our-th} in the context of the theory of weakly commutative triples (\ref{Kri1}). 

Note that the expansion (\ref{eq-ryad-ci-Schur}) 
gives the Puiseux series at infinity for the algebraic function 
defined by the polynomial relation $ Q (L_1, L_2) = 0 $
in the case of differential operators $ L_1, L_2 \in F [X] $, 
as follows from the Burchnall-Chaundy theorem. 
The Schur's method of the proof of Theorem~\ref{thSch-3} 
does not allow us to deduce this fact because of the obvious 
fundamental difficulty: the relations (\ref{eq-ryad-ci-Schur}) 
in the skew field $ \mathcal{F}_2 ((X)) $  are possible 
with arbitrary sets of constants $ c_k $, since 
$ \mathcal{F}_2 ((X)) $ includes not only the ``rational'' elements that form the Ore skew field $ \mathcal{F}_2 (X) $. 
Precisely for this reason the study of the smaller Ore skew field 
is important for our purposes. 

Another important result that leads naturally to the study of the skew field $ \mathcal{F}_2 (X) $ in the context of the theory of weakly commutative triples (\ref{Kri1}), will be given below in Section~\ref{sec-our-th}. 

\cite{Goodearl} gives a good overview of many useful algebraic constructions and results of the theory of Ore  skew fields and  skew fields of formal series of the form (\ref{pdo1}). 


\section{Simplest properties and examples of weakly commutative operators}\label{sec-primery} 

We consider algebraic properties of the set of operators $ L $ weakly commutative with a fixed $ H $. 
First, as we will see from the examples below (Examples~2 and 3), two such operators $ L_1 $, $ L_2 $ do not necessarily satisfy (\ref{Kri1}), which distinguishes our situation from the situation described in the 
Theorem~\ref{thSch-1}. However, their sums and products, as is easy to verify, still weakly commute with the given $ H $, thereby forming a subring of $ \mathcal{F}_2 [X, Y] $. 
If we now consider a more narrow set of operators, namely, the operators 
$ L_2 \in \mathcal{F}_2 [X, Y] $ for which (\ref{Kri1}) holds with the given 
$ L_1 $, $ H $, then again it is easy to see that all such $ L_2 $ are mutually $ H $-commutative.
This also follows from the results of Section~\ref{sec-our-th} and naturally leads to the 
consideration of the subfields of the skew field $ \mathcal{F}_2 (X) $ consisting of those elements of 
$ \mathcal{F}_2 (X ) $, that commute with its fixed element (the centralizer of the element 
in the algebraic terminology \cite{Goodearl}). 

\medskip 

We now consider several simplest transformations that simplify the form of weakly commutative triples of operators and the corresponding pairs of commuting elements of $ \mathcal{F}_2 (X) $ (see below Theorem~\ref{th2} in  Section~\ref{sec-our-th}). 
\begin{lemma}\label{lem-step0} 
  Suppose that the differential operators $ L $ weakly commute with $ H $: 
\begin{equation}\label{lem-st0-komm} 
[L, H] = D \cdot H, 
\end{equation}
$ D \in \mathcal{F}_2 [X, Y] $. 
If we replace $H$ with $ \widehat {H} = fH $, where $ f (x, y) \in \mathcal{F}_2 $,
then there exists an operator $ \widehat {D} \in \mathcal{F}_2 [X, Y] $ such that 
$$ [L, \widehat {H}] = \widehat {D} \cdot \widehat {H}, 
$$ 
so $ L $ and $ \widehat {H} $ commute modulo 
$ \widehat {H} $. 
\end{lemma} 
\begin{proof} 
  From (\ref{lem-st0-komm}) we conclude that 
  $ H \cdot L = (L-D) \cdot H $, so 
  $$ [L, \widehat {H}] = L \cdot \widehat {H} - \widehat {H} \cdot L = L \cdot \widehat {H} -fH \cdot L = L \cdot \widehat {H} -f (L-D) \cdot H. 
  $$ 
  Substituting $ H = f^{-1} \widehat {H} $ into this equation, we get the required 
  $$[L,\widehat{H}]=L\cdot \widehat{H}-f(L-D)f^{-1}\cdot \widehat{H} =(L-f(L-D)f^{-1})\cdot \widehat{H}.
  $$ 
  \end{proof}


It is natural to carry out the following simplification of the form of the operators $ L_i $ weakly commutative with $ H $. 
First, we note that since $ [P \cdot H, H] = [P, H] \cdot H $, 
$ P \cdot H $ weakly commutes with $ H $ for any operator $ P $, in particular, this is true for operators of the form $ a (x, y) X^{m} Y^{n} \cdot H $. 

We represent the general operator $L=\sum^{i+j \leq N}_{i,j=0}a_{ij}X^iY^j$ in the form 
$$L=L_{(1)}(X) + L_{(2)}(Y) + L_{(1,2)}(X,Y) + a_{00}(x,y),
$$
where 
$$L_{(1)}=\sum^n_{i=1}a_{i0}X^i, 
\quad L_{(2)}=\sum^m_{j=1}a_{0j}Y^j,
\quad L_{(1,2)}=\sum^{i+j \leq N}_{i,j=1}a_{ij}X^iY^j,
$$  
$ a_{ij} \in \mathcal{F}_2$, $m\leq N$, $n\leq N$. 

Ordering the mixed $ X, Y $-monomials in $ L_{(1,2)} $ lexicographically, take the greatest  
$ a_{r, s} X^{r} Y^{s} $. 
Then $ \widetilde {L} = L-a_{r, s} X^{r-1} Y^{s-1} \cdot H $ also weakly commutes with $ H $, 
but $ \widetilde {L} $ has a smaller leading monomial (after this operation 
$ L_{(1)} $, $ L_{(2)} $ may possibly change).
Thus we will consecutively remove mixed monomials in $ L $. In a finite number of steps, the operator $ L $ will take the following simple form 
\begin{equation}\label{upros_L1_L2}
  {\widehat{L}}=\sum^n_{k=1}p_k(x,y)X^k+\sum^m_{k=1}q_k(x,y)Y^k+p_{00}(x,y),
\end{equation} 
while $ \widehat {L} $ still weakly commutes with $ H $, and if the original $ L $ weakly commutes with some $ L_2 $, then it is easy to verify that $ \widehat {L} $ preserves this property. 

In what follows, we call {\it the order of the operator} $ L $, simplified 
as (\ref{upros_L1_L2}), the pair $ (n, m) $: $ {\rm ord}\ L = (n, m) $, and the simplified form 
(\ref{upros_L1_L2}) will be called 
{\it $H$-reduced form of $L$}.

Of course, reduction to the form (\ref{upros_L1_L2}) is also useful for checking conditions 
$[L_1, H] = 0 \ ({\rm mod}\ H)$. 

\medskip 

{\it Simplification of operators by a change of the independent variables $ x $, $ y $.} \\
Let us trace how the leading coefficients of the operators change under the substitutions $\hat{x}=\varphi(x)$, $\hat{y}=\psi(y)$.
Using the chain rule for the derivatives of function compositions, 
we obtain the corresponding transformations of the 
operators $ X^m $ and $ Y^n $ (we show only the leading terms): 
$$ X^m = \left(\varphi^{(1 )} \right)^m \hat{X}^m + \ldots, Y^n = \left(\psi^{(1)} \right)^n \hat{Y^n} + \ldots, 
$$ 
where 
$$X=\frac{\partial}{\partial x}, \quad \hat{X}=\frac{\partial}{\partial
\hat{x}}, \quad Y=\frac{\partial}{\partial
y}, \quad \hat{Y}=\frac{\partial}{\partial
\hat{y}}, \quad \varphi^{(1)}=\frac{\partial\varphi}{\partial
x}, \quad \psi^{(1)}=\frac{\partial\psi}{\partial y}.
$$ 
Let the original operator $ L $ be already $ H $-reduced. 
After the change of $x,y$, its leading coefficients, respectively, will be equal to 
$$\hat p_n(\hat{x})=p_n\left(\varphi^{-1}(\hat{x})\right) \cdot 
\left(\varphi^{(1)}\right)^n, \qquad
\hat q_m(\hat{y}) = q_m\left(\psi^{-1}(\hat{y})\right) \cdot \left(\psi^{(1)}\right)^m.
$$
The operator $ H = -XY + u (x, y) $ under the change of the variables $ x $, $ y $ becomes 
$\hat{H}=\varphi^{(1)}\psi^{(1)}\hat{X}\hat{Y}+
u\left(\varphi^{-1}(\hat{x}),\psi^{-1}(\hat{y})\right)$.


It is easy to show (cf.\ Theorem~\ref{th-step1} below) that
the leading coefficients $ p_n $, $ q_m $ of the operator $ L $ of the form (\ref{upros_L1_L2}) 
weakly commuting with $ H $ must depend only on one corresponding variable: $ p_n = p_n (x) $, $ q_m = q_m (y) $. 
Summarizing, we can state that for any operator $ L $ 
weakly commutative with $ H $, or for a weakly commutative triple (\ref{Kri1}), 
we can consider these operators to be $ H $-reduced (\ref{upros_L1_L2}), and changing the variables $ x $ , $ y $, we can bring the leading coefficients $ p_n $, $ q_m $ of 
one of the $ H $-commuting operators  to 1 (or, for even orders, if we do not want to resort to formally 
complex $ x $, $ y $, in some cases, to $-1$). 
Then, as it is easy to verify (see the beginning of the proof of Theorem~\ref{th-step1} below), for a weakly commutative triple the condition $ [L_1, L_2] = 0 \ ({\rm mod}\ H) $ implies that the leading coefficients of the second 
operator must be constant. 
In this case, of course, the operator $ H $ takes the form $ \tilde H = - \alpha (x) \beta (y) XY + \tilde u $. 
Dividing it by $ \alpha (x) \beta (y) $, we again obtain the form $ \hat H = -XY + \hat u $, keeping, by the Lemma~\ref{lem-step0}, the property of weak 
commutativity (\ref{Kri1}). 

We give below several examples of operators $ L $ commuting with $ H $ modulo $ H $. 

\medskip 

{\bf Example 1.} Let $ {\rm ord}\ L = (2,0) $, $ p_2 = 1 $, $ p_1 = 0 $, 
then the conditions of weak commutativity with $ H $ imply that $ p_{00} = \varphi (y) $ is an arbitrary function of $ y $; i.e.\ $ L = X^2 + \varphi (y) $. %
The remaining conditions of weak commutativity are 
$$\left \{
  \begin{array}{l}
  u_{xx}=0,\\
  u_y=-\frac{1}{2}\varphi_y. 
  \end {array}
  \right.
$$ 
This implies the following
form of the potential $ u $: 
$$ u (x, y) = - \frac{1}{2} \varphi'(y) x + \psi (y), 
$$ 
where $ \psi (y) $ is an arbitrary function of $ y $. 

\medskip 

{\bf Example 2.} Let  $ L $ of order $(2,0)$ have the 
general form $ L = x^2 + p_1 (x, y) x + p_{00} (x , y) $ with nonzero $ p_1 $. 
From the condition of weak commutativity we obtain a system of differential equations for 
$ p_1 $, $ p_{00} $, $ u $, from which it follows that 
$p_1=p_1(x)$, $p_{00}=p_{00}(y)$, $(p_{00})_y = -2u_x$, 
$u_{xx}=0$, $(p_1u)_x = 0$.
Finally, we get 
$p_1=\frac{c_2}{x-2c_1}$, $p_{00}=\varphi(y)$, $u=-\frac{1}{2}\varphi'(y)X+c_1\varphi'(y)$,
where $ c_1 $, $ c_2 $ are arbitrary constants and
$ \varphi (y) $ is an arbitrary function of $ y $. 

This example, despite its simplicity, gives one very useful result. Note that we have actually found {\it two} linearly independent operators that weakly commute with a fixed $H=-XY -\frac{1}{2}\varphi'(y)X+c_1\varphi'(y)$:
$$L_1=X^2+\varphi(y), \qquad L_2=\frac{1}{x-2c_1}X.
$$ 
If we calculate their commutator and reduce modulo $ H $, then we get a nonzero result. 

Thus unlike the situation in Theorem~\ref{thSch-1}, 
{\it if two operators commute modulo $ H $ with $ H $, 
then they do not necessarily commute with each other modulo $ H $.} 

One additional important observation is the fact that 
the condition $ p_1 \not \equiv 0 $, on one hand, 
implies the existence of {\it two} independent operators
that are weakly commutative with $H$
(but they do not  $ H $-commute with each other), 
on the other hand, this is {\it much more restrictive in the form of the potential} (compare Examples~1 and 2).


\medskip

We will devote a separate subsection to operators $ L $ of order $(2, 2)$. 

\subsection{Operators of order $(2,2)$ and addition theorems}

This case, on one hand, is also quite simple, 
on the other hand, it deserves detailed consideration, since it leads to one of the typical phenomena of algebro-geometric techniques in the 
theory of integrable nonlinear differential systems: 
non-trivial 
{\it addition theorems for special functions}.

We can again assume that the leading coefficients $ p_2 $, $ q_2 $ of the operators 
$ L_{(1)} (X) $ and $ L_{(2)} (Y) $ in the representation (\ref{upros_L1_L2}) are equal to 1. 
The case $ p_2 = 1 $, $ q_2 = -1 $ will be discussed below. 
The conditions for weak commutativity for an operator $ L $ of order $(2,2)$ with $ H $ can be easily found: 
\begin{enumerate}
  \item $(p_1)_y = (q_1)_x = 0$,
  \item $(p_{00})_y = -2u_x$,
  \item $(p_{00})_x = -2u_y$,
  \item $(p_{00})_{xy} + (p_1u)_x + (q_1u)_y + u_{xx} + u_{yy} = 0$.
\end{enumerate}
From the second and the third conditions we see that 
$ u_{xx} - u_{yy} = 0 $, 
i.e. 
\begin{equation}\label{vid_u}
  u = u_1\left(\frac{x+y}{2}\right) + u_2\left(\frac{x-y}{2}\right),
\end{equation}
which means separation of variables in the Schr\"odinger 
operator $ H = -XY + u $ in the coordinates 
$ \hat x = (x + y) / 2 $, $ \hat y = (x-y) /2$. 
This condition together with the other two 
leads to a special form of the potential $ u $. 
Namely, taking into account $ u_{xx} - u_{yy} = 0 $, we obtain from the four conditions above 
that it is necessary and sufficient that the equality $ (p_1u)_x + (q_1u)_y = 0 $ be satisfied. 
Introducing the function $ \psi (x, y) $ such that 
\begin{equation}\label{arr_u_psi} 
\psi_x = q_1u, \qquad \psi_y = - p_1u, 
\end{equation} 
we get the equivalent condition 
\begin{equation}\label{ur_psi} 
p_1 (x) \psi_x + q_1 (x) \psi_y = 0. 
\end{equation} 
Consider three possible cases: 
\begin{enumerate}
  \item $p_1 \equiv q_1 \equiv 0$;
  \item $p_1 \not\equiv 0 $, $ q_1 \equiv 0$ 
      (the case $p_1 \equiv 0 $, $ q_1 \not\equiv 0$
      is completely analogous);
  \item $p_1 \not\equiv 0 $, $ q_1 \not\equiv 0$.
\end{enumerate} 

{\bf The first case} is trivial: 
$ L = X^2 + Y^2 - 2 \left(u_1 \left(\frac{x + y} {2} \right) - u_2 \left(\frac{xy} {2} \right) \right) $ 
and in the rotated variables $ \hat x = (x + y) / 2 $, $ \hat y = (x-y) / 2 $ coincides up to the signs with $ \hat X $ and $ \hat Y $-parts of $ H $. 

\medskip 

{\bf The second case} gives a partial degeneration: 
$ (p_1u)_x = 0 $, that is, 
\begin{equation}\label{teo-slozh-1}
  u(x,y) =  u_1\left(\frac{x+y}{2}\right) + u_2\left(\frac{x-y}{2}\right) = \alpha(y)/p_1(x).
\end{equation}
It is easy to see that this identity is the simplest form of an {\it addition theorem}. Further consideration of this case will be postponed until the third case is completed. 

\medskip 

{\bf The third case}. The condition (\ref{ur_psi}) allows us to find the form of $ \psi $: 
\begin{equation}\label{vid_psi}
  \psi(x,y) = \phi(\alpha(x) - \beta(y))
\end{equation} 
for some functions $ \phi $, $ \alpha $, $ \beta $ of one variable, while $ p_1 (x) = 1 / \alpha'(x) $, $ q_1 (y) = 1 / \beta'(y) $. 

The conditions (\ref{arr_u_psi}) connect the functions 
$\psi$ and $u$, which have simple forms (\ref{vid_psi}), 
(\ref{vid_u}). 
Let us show a way to get an addition theorem from this fact.
First of all we see from (\ref{arr_u_psi}) that 
\begin{equation}\label{u_psi-shtrih}
  u = u_1((x+y)/{2}) + u_2((x-y)/{2}) = \alpha'(x)\beta'(y)\phi'(\alpha(x) - \beta(y)). 
\end{equation}
We integrate this equation w.r.t.\ $ x $ and $ y $, then 
introduce the primitives $ U_1 $, $ U_2 $, $ \Phi $ 
such that 
$U_1''(z) = u_1(z)$, 
$U_2''(z) = u_2(z)$, $\Phi'(z) = \phi(z)$.  
Then the final necessary and sufficient condition of 
weak commutativity of the operator $L$ in the third case consists in the identity 
\begin{equation}\label{teor-slozh-poln}
  \Phi(\alpha(x)-\beta(y)) = U_1\left(\frac{x+y}{2}\right) - U_2\left(\frac{x-y}{2}\right)+
  U_3(x)+U_4(y)
\end{equation}
for the specified 7 functions of one variable. 

{\it The problem of finding all possible non-trivial identities of the form (\ref{teor-slozh-poln}) is of obvious interest.} 

A number of examples of such sort 
is given by the addition theorems from different branches
of mathematics, starting with trivial algebraic ones: 
$$x^2-y^2=(x+y)(x-y),$$ 
$$\cos x+\cos y=2\cos\frac{x+y}{2}\cos\frac{x-y}{2},$$
$$\sinh x+\sinh y=2\sinh\frac{x+y}{2}\cosh\frac{x-y}{2},$$ 
and other similar trigonometric and exponential ones. 
After applying the logarithm they are reduced to the form (\ref{teor-slozh-poln}): 
$$\log(x^2-y^2)=\log\frac{x+y}{2}+\log\frac{x-y}{2}+\log 4,
$$
$$\log(\cos x+\cos y)=\log\cos\frac{x+y}{2}+\log\cos\frac{x-y}{2}+\log 2.
$$ 

In the latter case, we obtain the following values of the 
coefficients of the operator $ L = X^2 + p_1X + Y^2 + q_1Y + p_{00} $ and the potential $ u $: 

{\bf Example 3.} 
$$p_1=\frac{c_1}{\sin x}, \qquad q_1=-\frac{c_1}{\sin y}, \qquad p_{00}=-2\left (\frac{1}{\cos^2
(\frac{x+y}{2})} +\frac{1}{\cos^2(\frac{x-y}{2})} \right ),
$$
$$u=\frac{1}{\cos^2(\frac{x+y}{2})}-\frac{1}{\cos^2(\frac{x-y}{2})}.
$$

A rich stock of such identities is given by the theory of special functions.
For example the doubly periodic Jacobi function 
$ {\rm sn} \, z $ 
satisfy the following identity (cf. \cite{WWcma}): 
\begin{equation}\label{slozh-sn}
  {\rm sn}(x+y)\,{\rm sn}(x-y)=\frac{{\rm sn}^2x-{\rm sn}^2y}{1-k^2{\rm sn}^2x\,{\rm sn}^2y}.
\end{equation} 
Multiply both its parts by $ k $ and  
define the functions $ \alpha (x) $, $ \beta (y) $ by the 
relations 
$k\,{\rm sn}^2x={\rm tanh}\,\alpha(x)$, $k\,{\rm sn}^2y={\rm tanh}\,\beta(y)$.
Then the right hand side of (\ref{slozh-sn}) may be represented as the formula for the hyperbolic tangent of the 
difference of the arguments: 
$$\frac{{\rm tanh}\,\alpha(x)-{\rm tanh}\,\beta(y)}{1-{\rm tanh}\,\alpha(x){\rm tanh}\,\beta(y)}=
{\rm tanh}(\alpha(x)-\beta(y)),
$$ 
i.e.\ we obtain 
$$k\,{\rm sn}(x+y)\,{\rm sn}(x-y)={\rm tanh}(\alpha(x)-\beta(y)).
$$  
Again applying the logarithm we have: 
$$\log {\rm tanh}(\alpha(x)-\beta(y))=\log
{\rm sn}(x+y)+\log {\rm sn}(x-y)+ \log k,
$$ 
where $\alpha(z)=\beta(z)={\rm arcth}(k\,{\rm sn}^2(z)).$ 
that is the necessary identity of the form
(\ref{teor-slozh-poln}).

In the same way one may reduce the identities 
$$\wp(u) - \wp(v) = - \frac{\sigma(u+v)\sigma(u-v)}{\sigma^2(u)\sigma^2(v)},
$$
$$\wp(u+v) - \wp(u-v) = - \frac{\wp'(u)\wp'(v)}{(\wp(u) - \wp(v))^2}.
$$ 
to the form (\ref{teor-slozh-poln}) 
or (\ref{u_psi-shtrih}).

We also give here a more complicated relation for 
the Jacobi
$ \theta $-functions (see \cite{WWcma}), which can also be reduced to the form (\ref{teor-slozh-poln}): 
$$\theta_3(z+y)\theta_3(z-y)\theta^2_3=\theta^2_3(y)\theta^2_3(z)+
\theta^2_1(z)\theta^2_1(y).$$
We rewrite it in the form 
$$\frac{\theta^2_1(z)}{\theta^2_3(z)}+
\frac{\theta^2_3(y)}{\theta^2_1(y)}=\frac{\theta_3(z+y)\theta_3(z-y)\theta^2_3}
{\theta_3^2(z)\theta_1^2(y)}
$$
and denoting 
$$ \frac{\theta^2_1 (z)} {\theta^2_3 (z)} = \alpha (z), \qquad \frac{\theta^2_3 (y)} {\theta^2_1 (y)} =
 - \beta (y), 
$$ 
we apply the logarithm to the both parts: 
$$\log(\alpha(z)-\beta(y))=
\log\theta_3(z+y)+\log\theta_3(z-y)-\log\theta_3^2(z)-\log(\theta_1^2(y)/\theta_3^2).
$$  
This again gives (\ref{teor-slozh-poln}). 

A large stock of addition theorems of the form 
(\ref{teor-slozh-poln}) (sometimes called addition pseudo-theorems because of the presence of not only $ x + y $, but also the difference $ x-y $) 
can be found in \cite{WWcma} and the original 
papers by K.~Jacobi. 

Returning to the discussion of the equality 
(\ref{teo-slozh-1}) for the second case, 
we can state that there is a large set of non-trivial examples of similar identities in which both elementary and special functions can enter. 
Each identity of the form (\ref{teo-slozh-1}), 
(\ref{u_psi-shtrih}) or (\ref{teor-slozh-poln}) 
gives an example of {\it two} operators $ L $ weakly commutative with $ H $ (but, as a rule, not $ H $-commuting among themselves). 
Some of these examples are trivialized after transformation to the variables $ \hat x = (x + y) / 2 $, $ \hat y = (x-y) / 2 $. 
However, it should be noted that, in spite of the fact that the Schr\"odinger operator admits 
separation of variables in all such cases, some of the cases considered above are of some 
interest for the theory of weakly commutative triples. Example~3 above contains already 
{\it two} operators weakly commutative with $ H $, one of order 
$(2,2)$ (in which we can put $ c_1 = 0 $) and one of 
order $(1,1)$, the existence of the latter already imposes very strict restriction (\ref{ur_psi}), 
equivalent to the addition theorem. 
Moreover, as in Example~2, these two operators will not be 
mutually $ H $-commuting. 

Let us briefly consider the case $ p_2 = 1 $, $ q_2 = -1 $. As is easy to verify, the conditions of weak
commutativity give 
$ u_{xx} + u_{yy} = 0 $, i.e. the potential $ u $ must be a harmonic function. 
The remaining conditions give 
$p_1=p_1(x)$, $ q_1=q_1(y)$, 
$(p_{00})_y = -2u_x$,
$(p_{00})_x = 2u_y$,
$(p_{00})_{xy} + (p_1u)_x + (q_1u)_y + u_{xx} - u_{yy} = 0$.  

Again we get the equality $ (p_1u)_x + (q_1u)_y = 0 $ 
and after the introduction of $ \psi (x, y) $ such that (\ref{arr_u_psi}) holds and the corresponding second antiderivatives are introduced, we get
(for the maximally non-degenerate case, $ p_1 \not \equiv 0 $, $ q_1 \not \equiv 0 $) 
\begin{equation}\label{teor-slozh-garmon}
  U(x,y) = \Phi(\alpha(x)-\beta(y)) +  U_3(x)+U_4(y)
\end{equation}
for a harmonic function $ U (x, y) $. Such a form of {\it harmonic addition theorem} certainly deserves a separate study. 
As above, from each non-trivial identity of the form 
(\ref{teor-slozh-garmon}) 
one can construct an example of a pair of operators weakly commuting with $ H = -XY + u (x, y) $, 
without any obvious separation of variables (in the real case). 


\section{Basic theorems}\label{sec-our-th} 

First of all, we note that when we consider the properties of the triples of $ H $-commuting operators 
(\ref{Kri1}), using the skew field $ \mathcal{F}_2 (\partial_x) = \mathcal{F}_2 (X) $ and the ring $ \mathcal{F}_2 (\partial_x) [\partial_y] = \mathcal{F}_2 (X) [Y] $ leads to much simpler formulations. 

Denote as $ M $ and $ \widetilde {M} $ the following formal operators from $ \mathcal{F}_2 (X) [Y] $ and 
$ \mathcal {F}_2 (Y) [X] $, respectively: 
\begin{equation}\label{MMM}
  M = - X^{-1}\cdot H = Y -  X^{-1}\cdot u,
  \qquad \widetilde M = - Y^{-1}\cdot \widetilde{H} =
  X - Y^{-1}\cdot u,
\end{equation} 
where, as we mentioned in the introduction, we use the hyperbolic form of the operator $ H = - \partial_x \partial_y + u = - XY + u $. 

Suppose that $ H $ and the operators $ L_1 $ and $ L_2 $ satisfy (\ref{Kri1}). 
Dividing with remainder the operators $ L_1 $ and $ L_2 $ by the operator $ M $ (or $ \widetilde {M} $) in 
the ring $ \mathcal{F}_2 (X) [Y] $ 
(resp. $ \mathcal{F}_2 (Y) [X] $), we obtain formal linear \emph{ordinary} pseudodifferential operators 
$ R_1 $, $ R_2 $ which belong to 
the Ore skew field $ \mathcal{F}_2 (X) $ such that 
\begin{equation}\label{L2R} 
  L_1 - Q_1\cdot M=R_1 \in \mathcal{F}_2(X), 
  \quad L_2 -
  Q_2\cdot M=R_2 \in \mathcal{F}_2(X)
\end{equation} 
with $ Q_i \in \mathcal{F}_2(X)[Y]$. 

The following simple theorem was proved in 
\cite{TT10}: 
\begin{theorem}\label{th2} 
For the operators $ R_1 $, $ R_2 $ and $ M $, the following relations hold:
\begin{enumerate} 
\item $[R_1,R_2] = 0 $ in $ \mathcal{F}_2 (X) $; 
\item $[R_1,M]=[R_2, M]=0$ in $\mathcal{F}_2 (X)[Y]$; 
\item if there exists a polynomial $ Q (L_1, L_2) $ with constant coefficients, such that 
$ Q (L_1, L_2) = 0\ (mod\ H)$ in 
$ \mathcal{F}_2[X,Y]$, 
then $ Q (R_1, R_2) = 0 $ in $ \mathcal{F}_2 (X) $. 
\end{enumerate} 
\end{theorem} 
So obviously the theory of commuting elements of the Ore skew field $ \mathcal{F}_2 (X) $  
is of great importance for the theory of integrable 
nonlinear partial differential equations. It resembles in many aspects the Burchnall-Chaundy theory of commuting ordinary differential operators, although 
not all of the classical results of Burchnall and Chaundy have been carried over 
to the case under consideration. We refer again to the 
review \cite{Goodearl} and recent work by other authors \cite{Richter-Silv}. 

First we prove several results for the ring 
$ \mathcal{F}_2 ((X)) [Y] $. 
Consider the operator 
$$
M=Y-X^{-1}u=Y+\sum^\infty_{i-1}(-1)^iu^{i-1,0}X^{-i},
$$ 
(here and below we use the notation
 $ a^{i, j} \equiv \frac{\partial^{i + j} a} {\partial x^i \partial y^j} $ 
for the partial derivatives of functions 
$ a (x, y) \in \mathcal{F}_2 $) and two elements of the skew field $ \mathcal{F}_2 ((X)) $ of the form 
$$ R_1 = \sum^\infty_{i = 0} r_{1-i} X^{1-i} 
$$ 
(thus $ R_1 $ is of order 1) and 
$$\widetilde{R}=\sum^\infty_{i=0}p_{k-i}X^{k-i}
$$ 
of arbitrary order $ k $. 
\begin{theorem}\label{th-step1} 
If the operators $ M $, $ R_1 $, $ \widetilde {R} $ pairwise commute, 
that is, $ [M, R_1] = [M, \widetilde {R}] = [R_1, \widetilde {R}] = 0 $, 
then there exist constants $ c_0 $, $ c_1 $, $ c_2, \ldots, $ such that 
\begin{equation}\label{eq-ryad-ci} 
\widetilde {R} = \sum^\infty_{i = 0} c_iR_1^{k-i}. 
\end{equation} 
\end{theorem} 
\begin{proof} 
(We follow here the scheme of I.Schur~\cite{Schur}). 
From 
$$ [M, R_1] = r_1^{0,1} X^1 + r_0^{0,1} + (r_{- 1}^{0,1} + r_1^{1,0} u + r_1u^{1,0}) X^{- 1} + \ldots = 0 
$$ 
one concludes that $ r_1^{0,1} = 0 $, 
$ r_0^{0,1} = 0, $ so 
$ r_1 $ and $ r_0 $ are functions of $ x $ only: 
$ r_1 = r_1 (x) $, $ r_0 = r_0 (x) $. 
Similarly, $ [M, \widetilde {R}] = 0 $ implies, that 
$ p_{k}^{0,1} = p_{k-1}^{0,1} = 0 $, that is, 
$ p_{k} (x) $ and $ p_{k-1} (x) $ also does not dependent on $ y $. 
From the commutativity of the operators $ R_1 $ and 
$ \widetilde {R} $, we have 
$ r_1p_{k}^{1,0} - kr_1^{1,0} p_{k} = 0 $, 
thus there exists a constant $ c_0 $ such 
that $ p_{k} = c_0r_1^{k} $, 
since $ r_1 $ and $ p_{k} $ depend only on $ x $. 

Hence, the operator $ \widetilde {R} $ of order $ k $ is representable as 
$$\widetilde{R}=c_0r_1^{k}X^{k}+p_{k-1}X^{k-1}+\ldots.
$$
We subtract from $ \widetilde {R} $ the operator $ c_0R_1^{k} $ obtaining the new operator 
$ \widetilde {\widetilde {R}} = \widetilde {R} -c_0R_1^{k} $ of order $ k-1 $. 
Obviously, the operator 
$ \widetilde {\widetilde {R}} $ also commutes with $ R_1 $ and $ M $. 

Proceeding further we find a constant $ c_1 $ such 
that $ \widetilde {\widetilde {R}} = c_1r_1^{-k-1} X^{-k-1} + \ldots. $ 
This results in 
$$ \widetilde {R} = c_0R_1^{-k} + c_1R_1^{-k-1} + \ldots $$ 
with a remainder of order $ k-2 $. 
Continuing this procedure, we obtain the required 
decomposition (\ref{eq-ryad-ci}).
\end{proof}


\begin{lemma}\label{lem-step11} 
If $ P \in \mathcal{F}_2 ((X)) $ of order $ n $ commutes with $ M $, 
that is, $ [P, M] = 0 $ in the ring 
$ \mathcal{F}_2 (X) [Y] $, 
then $ R = \sqrt [n] {P} $ 
(the root of degree $ n $ of the operator $ P $ ) 
also commutes with $ M $: $ [R, M] = 0 $. 
\end{lemma} 
\begin{proof} 
The proof of the existence of a root of degree $ n $ of any element of $ \mathcal{F}_2 ((X)) $ of order 
$ n $ was given by Schur (Theorem~\ref{thSch-2}). 
Consider the following chain of obvious equalities 
\begin{equation}\label{kommPM} 
0=[P,M]=[R^n,M]=\sum^n_{i-1}R^{n-i}[R,M]R^{i-1}.
\end{equation} 
Then, proving by contradiction, suppose that $ [R, M] \neq 0 $, where 
$$R=\sqrt[n]{P}=r_1X^1+r_0X^0+r_{-1}X^{-1}+\ldots,
$$
$M=Y-uX^{-1}+u^{1,0}X^{-2}-u^{2,0}X^{-3}+\ldots.$ 
Let the leading term of the commutator $ [R, M] $ 
equal 
$ \phi_s (x, y) X^s $, then the leading term in the sum (\ref{kommPM}) is $nr_1^{n-1}\phi_s X^{n-1+s}$, 
i.e. the sum is not equal to zero, which contradicts to (\ref{kommPM}). 
\end{proof}

\begin{theorem}\label{thTs-BCh} 
  For two
operators $ L_1 $, $ L_2, \in \mathcal{F}_2 [X, Y] $
and $ H = - \partial_x \partial_y + u $ 
satisfying (\ref{Kri1}), 
there exists a polynomial $ Q (L_1, L_2) $ 
with constant coefficients such that 
$ Q (L_1, L_2) = S \cdot H $, 
$ S \in \mathcal{F}_2 [X, Y] $. 
\end{theorem} 
\begin{proof} 
We assume that the operators $ L_1 $, $ L_2, \in \mathcal{F}_2 [X, Y] $ have been 
simplified to the form (\ref{upros_L1_L2}):
$$L_1 = L_{1{(1)}}(X) +  L_{1{(2)}}(Y) + p_{100}(x,y) 
= \sum^{n_1}_{k=1}p_{1k}(x,y)X^k+\sum^{n_2}_{k=1}q_{1k}(x,y)Y^k + p_{100}(x,y),
$$
$$L_2 = L_{2{(1)}}(X) +  L_{2{(2)}}(Y) + p_{200}(x,y) 
= \sum^{m_1}_{k=1}p_{2k}(x,y)X^k+\sum^{m_2}_{k=1}q_{2k}(x,y)Y^k + p_{200}(x,y).
$$
Then the transformation (\ref{L2R}) of the operators 
$ L_1 $, $ L_2 $ into the commuting pseudodifferential 
operators $ R_1, R_2 \in \mathcal{F}_2 (X) $ 
preserves their leading coefficients $ p_{1n_1} $
and $ p_{2m_1} $. 
As follows from the proof of  
Theorem~\ref{th-step1} 
(and Theorem~\ref{thSch-3}), 
these leading coefficients thus have the form 
$p_{1n_1}= c_1(\phi(x))^{n_1}$, 
$p_{2m_1} = c_2(\phi(x))^{m_1}$ 
for some function $ \phi (x) $ of one variable $ x $. Applying the same transformation (\ref{L2R}) to the 
commuting pseudodifferential operators $ \widetilde R_1, \widetilde R_2 \in \mathcal{F}_2 (Y) $ 
using the $ \widetilde M $ operator, we find that the leading coefficients $ q_{1n_2} $ and $ q_{2m_2} $ have the form $ d_1 (\psi (y))^{n_2} $, 
$ d_2 (\psi (y))^{m_2} $. 

Consider the monomials of a general polynomial with constant coefficients $ Q (L_1, L_2) = \sum_{r, s}^{r + s \leq N} b_{r, s} (L_1)^r (L_2)^s $. 
Since the operators $ L_1 $, $ L_2 $ $ H $-commute, 
we suppose their order inside each monomial to be fixed (first stands a power of $ L_1 $, then 
a power of $ L_2 $). 
$ H $-reducing each term 
$$ b_{r, s} (L_1)^r (L_2)^s = P_{r, s} H + W, 
$$
we obtain the following operator $ H $-commuting with $ L_1 $, $ L_2 $:
\begin{equation}\label{H-upr_w} 
W = W_{(1)} (X) + W_{(2)} (Y) + w_{00} (x, y), 
\end{equation} 
whose order is bounded above by 
$ N \max (n_i, m_i) $. The leading coefficients of 
$ W_{(1)} (X) $, $ W_{(2)} (Y) $ should have the form 
\begin{equation}\label{starsh-koeff} 
c_3 (\phi (x))^{n_3}, \qquad c_4 (\psi (y))^{m_2}, \qquad c_i = {\rm const} 
\end{equation} 
respectively (the other coefficients are significantly 
more complicated). 
The complete $ H $-reduced operator $ Q (L_1, L_2) $ 
also weakly commutes with $ L_1 $, $ L_2 $, $ H $, hence its leading coefficients must also have the form 
(\ref{starsh-koeff}) with constants $ c_i $, linear in $ b_{r, s} $. 
Thus, equating these two leading coefficients to zero, 
we obtain two linear equations on $ b_{r, s} $ of the form 
\begin{equation}\label{lin_c_cherez_b} 
\sum_{r + s \leq N} \lambda_{r, s } b_{r, s} = 0. 
\end{equation} 
Assuming that they hold, we see that the $ H $-reduced 
form of $ Q (L_1, L_2) $ in representation (\ref{H-upr_w}) has orders at least by 1 smaller than the initial one, under two linear conditions on the coefficients $ b_{r, s} $. 
Moreover, the leading coefficients of the resulting operator due to its $ H $-commutativity with $ L_1 $, 
$ L_2 $, $ H $ also have the form (\ref{starsh-koeff});
again equating them to zero, we obtain two more linear equations of the form (\ref{lin_c_cherez_b}). 
Obviously, for sufficiently large $ N $, the number of 
terms (approximately $ N^2/2 $) in $ Q (L_1, L_2) $ will be greater than the number of equations 
(\ref{lin_c_cherez_b})
(approximately $ 2N \max (n_i, m_i) $) which
guarantee vanishing of all the coefficients of 
the $ H $-reduced operator $ Q (L_1, L_2) $ in the form (\ref{H-upr_w}). 
The remaining $ w_{00} (x, y) $ must weakly commute with $ H $, so
is a constant. 
Thus, the existence of a nontrivial polynomial $ Q (L_1, L_2) $ with constant coefficients such that $ Q (L_1, L_2) = 0 \ ({\rm mod}\ H) $ is obvious. 
\end{proof}


\section{Conclusion}\label{zakl} 

The review \cite{Goodearl} contains many results of the theory
of Ore skew fields, their various algebraic 
generalizations (as well as for 
skew fields of formal series of the form 
(\ref{pdo1})), some generalizations of the Burchnall-Chaundy theorem on polynomial relations among commuting elements, commutativity and the finite-dimensionality of the centralizers of 
non-trivial elements (generalizations of Theorem~\ref{thSch-1} by I.Schur). 
In particular, \cite{Goodearl} refers to a theorem of Resco, Small and Wadsworth (Theorem~5.13) 
on the upper bound on the degree of transcendence of any commutative subfield of a (generalized) Ore skew field 
$ R (\delta) $, which, in particular, implies a 
theorem of Burchnall-Chaundy type for the case when 
the field of coefficients of the original ring of ordinary differential operators is a field of rational 
or algebraic functions of one variable. However, it 
seems that until now the general theorem on the algebraic dependence of any two commuting elements
(the analogue of Theorem~\ref{thTs-BCh}) for the general case of one-dimensional Ore skew field
$ F (X) $ with an arbitrary differential field of the coefficients $ F $ of characteristic 0 is unknown. 
Note that the commuting elements $ R_1 $, $ R_2 \in \mathcal{F}_2 (X) $ corresponding to the operators $ L_1 $, $ L_2 \in \mathcal{F}_2 [X, Y] $ in 
Theorem~\ref{thTs-BCh}, have a very special form. 
It is of great interest either to prove a similar general algebraic theorem for $ F (X) $, or to give a counterexample. 
As is obvious from Theorem~\ref{thSch-3}, 
it may be easy to prove only the existence of a
{\it functional} (possibly, transcendental) relation between 
commuting elements $ R_1 $, $ R_2 \in F (X) $, 
provided the Puiseux (\ref{eq-ryad-ci-Schur}) expansion converges. 


\end{document}